\documentclass[journal,cls,onecolumn,12pt,twoside]{IEEEtran}
\usepackage{etex}
\normalsize
\usepackage[T1]{fontenc}
\usepackage{amsmath,amssymb,amsfonts}
\usepackage{mathtools}
\usepackage{mathrsfs}
\usepackage{mathabx}
\usepackage{amsbsy}
\usepackage{graphicx}
\usepackage{epstopdf}
 \usepackage{algorithm}
\usepackage[algo2e]{algorithm2e} \usepackage{subfigure}
\usepackage{cite}
\usepackage{multirow}
\usepackage{ctable}
\usepackage{caption}
\usepackage{setspace}
\usepackage{mathtools}
\usepackage{hyperref}

\newcommand\myeq{\stackrel{\mathclap{\normalfont\mbox{def}}}{=}}

\title{\Huge$\,$\\[-2.75ex]
{A New Algebraic Approach for String Reconstruction from Substring Compositions}\\[0.50ex]}

\author{
\IEEEauthorblockN{
    Utkarsh Gupta and
	Hessam Mahdavifar
	}\\
\IEEEauthorblockA{\small 
Department of Electrical Engineering and Computer Science, University of Michigan, 
Ann Arbor, MI 48109, USA
} \\
\IEEEauthorblockA{\small Emails: utkarshg@umich.edu,  hessam@umich.edu}\\[-2mm]
\vspace{-3mm}

\thanks{ An earlier version of this paper was presented in part at the 2022 IEEE International Symposium on Information Theory [DOI: 10.1109/ISIT50566.2022.9834531]. This work was supported by the National Science Foundation under grant CCF--1941633.}
}

\usepackage{mathtools}
\usepackage{amsthm}
\newtheorem{theorem}{{Theorem}}
\newtheorem{lemma}[theorem]{{Lemma}}
\newtheorem{proposition}[theorem]{{Proposition}}
\newtheorem{corollary}[theorem]{{Corollary}}
\newtheorem{remark}{Remark}

\newtheorem{definition}{{Definition}}

\DeclareMathAlphabet{\mathbfsl}{OT1}{ppl}{b}{it} 

\newcommand*{\rom}[1]{\expandafter\romannumeral #1}

\def\proof{\noindent\hspace{2em}{\itshape Proof: }}

\makeatletter
\newcommand{\AlignFootnote}[1]{%
  \ifmeasuring@
  \else
    \iffirstchoice@
      \footnote{#1}%
    \fi
  \fi}
\makeatother

\newcommand{\be}[1]{\begin{equation}\label{#1}}
\newcommand{\ee}{\end{equation}}

\renewcommand{\le}{\leqslant} 

\renewcommand{\ge}{\geqslant}

\renewcommand{\Bbb}{\mathbb}

\newcommand{\Cref}[1]{Co\-ro\-lla\-ry\,\ref{#1}}

\newcommand{\Fn}{\smash{\Bbb{F}_{\!2}^{\hspace{1pt}n}}}

\renewcommand{\labelenumi}{\theenumi}

\begin{document}

\maketitle
\thispagestyle{empty} 
\pagestyle{plain}    

\begin{abstract}
In this paper, we propose a new algorithm for the problem of string reconstruction from its substring composition multiset. Motivated by applications in polymer-based data storage for recovering strings from tandem mass-spectrometry sequencing, the algorithm exploits the equivalent polynomial formulation of the problem. We characterize sufficient conditions for a length $n$ binary string that guarantee the string's reconstruction time complexity to be bounded polynomially as $O(n^2)$. This improves the time complexity of the reconstruction process compared to the $O(n^2 \log{n})$ complexity of the algorithm by Acharya \textit{et al}. for this problem \cite{acharya2015string}. Moreover, the sufficient conditions on binary strings that guarantee reconstruction in polynomial time are more general than the conditions for the algorithm by Acharya \textit{et al}. This is used to construct new codebooks of \textit{reconstruction codes} that have efficient encoding procedures, and are larger, by at least a linear factor in size, compared to the previously best known construction by Pattabiraman \textit{et al.} \cite{pattabiraman2023coding}. 
\end{abstract}

\section{Introduction}
Recent years have seen an explosion in the amount of data created globally \cite{rydning2018digitization}. The volume of data generated, consumed, copied, and stored is projected to reach more than 180 zettabytes by 2025. In 2020, the total amount of data generated and consumed was 64.2 zettabytes \cite{statistaData}. But traditional digital data storage technologies such as SSDs, hard drives, and magnetic tapes are approaching their fundamental density limits and would not be able to keep up with the increasing memory needs \cite{hilbert2011world}. 

Several molecular paradigms with significantly higher storage densities have been proposed recently \cite{ng2021data,launay2021precise, dickinson2021alternative, dahlhauser2021efficient, matange2021dna, rutten2018encoding, al2017mass, erlich2017dna, zhirnov2016nucleic, grass2015robust, yazdi2015rewritable, goldman2013towards}. Molecules with a structure consisting of different smaller molecules (monomers) joined together in sequences are called polymers. If different types of molecules denote different alphabets, then a polymer with a linear arrangement of these molecules, i.e., a polymer string, can be treated as a sequence of alphabets. DNA is one promising data storage medium which has generated significant interest in the data storage research community. However, DNA has several scalability constraints including an expensive synthesis and sequencing process which prevent large-scale commercialization. Furthermore, DNA is prone to a diverse type of errors such as mutations within strands, or loss of strands due to breakage or degradation that could lead to potential decoding errors or even complete loss of information \cite{matange2021dna}. 



 This has led researchers to search for alternatives in other synthetic polymers. For example, synthetic proteins (which are polymers of amino acids) are emerging as a potential alternative with data being stored using peptide sequences for the first time in 2021 \cite{ng2021data}. Compared to DNA and other types of polymers, proteins offer several advantages for data storage, including higher stability of some proteins than DNA \cite{warren2019move}, and availability of a larger set of possible monomers ($20$ amino acids are observed in natural proteins).  
 
In synthetic polymer strings, monomer units of different masses, which represent the two bits $0$ and $1$, are assembled into user-determined readable sequences.  A common family of technological methods for reading amino-acid sequences (and other bio-polymers) is mass spectrometry \cite{creighton1993proteins}. Mass spectrometers take a large number of identical polymer strings, randomly break the polymer into substrings, and analyze the resulting mixture. The mass sequencing spectrum obtained gives us the mass and frequency of each contiguous molecular substring. The process of recovering a molecular string from its mass sequencing spectrum is modeled into the problem of reconstructing a string from the multiset of the compositions of its contiguous substrings. 

The class of problems of reconstructing a string from substring information usually falls under the general framework of the string reconstruction problems. Due to their relevance in modelling molecular storage frameworks, the list of recent work in string reconstruction problems has grown rapidly \cite{9950539,ye2022reconstruction,marcovich2021reconstruction, gabrys2021reconstructing,cheraghchi2020coded, abroshan2019coding, gabrys2018unique, kiah2016codes,acharya2015string, motahari2013information}. The problem of string reconstruction from its substring compositions was first introduced in \cite{isit2010} and \cite{acharya2015string}. The main results from \cite{acharya2015string} assert that binary strings of length $\le 7$, one less than a prime, and one less than twice a prime are uniquely reconstructable, from their \textit{substring composition multiset}, up to reversal. The authors of \cite{acharya2015string} also introduced a backtracking algorithm for reconstructing a binary string from its \textit{substring composition multiset}, and provide sufficient conditions for reconstructability of a binary string using the proposed algorithm in \cite{acharya2015string} without the need for backtracking (lemma~\ref{property}). In the case of no backtracking, this algorithm has a time complexity of $O(n^2 \log{n})$. Note that in the case of backtracking, there is no guarantee that the time complexity will remain bounded polynomially with $n$. Relying on this reconstruction algorithm, the works of \cite{pattabiraman2023coding}, \cite{pattabiraman2019reconstruction} and \cite{gabrys2020mass} viewed the problem from a coding theoretic perspective. They proposed coding schemes that are capable of correcting a single mass error and multiple mass errors, respectively, and can be reconstructed by the reconstruction algorithm without backtracking.

The problem formulation in \cite{acharya2015string}, and subsequently in\cite{pattabiraman2023coding}, relies on the two following assumptions: a) One can uniquely infer the composition (number of monomers of each type) of a polymer from its mass; and b) The masses of all the substrings of a polynomial are observed with identical frequencies. In this work, we also rely on these assumptions.

In the context of combinatorics, the problem is closely related to the turnpike problem, also known as the partial digest problem, where the locations of $n$ highway exits need to be recovered from the multiset of their $\binom{n}{2}$ interexit distances. In \cite{acharya2015string}, the authors showed that the problem of string reconstruction from its \textit{composition multiset} can be reduced to an instance of the turnpike problem.

In this paper, we propose a new algorithm to reconstruct the set of binary strings with a given multiset of substring compositions. The proposed algorithm relies on on the algebraic properties of the equivalent bivariate polynomial formulation \cite{acharya2015string} of the problem. The algorithm finds the coefficients of the corresponding polynomial in a manner that reconstructs the binary string from both ends progressing towards the center. We show that the time complexity of the reconstruction process is reduced with our proposed algorithm compared to the combinatorial algorithm proposed in \cite{acharya2015string}. However, in general, a drawback of such algorithms is that they may need backtracking which can lead to reconstruction complexity that grows exponentially with the length $n$, in a worst case scenario. Therefore, we provide algebraic conditions on binary strings that are sufficient to guarantee unique reconstruction by the proposed algorithm without backtracking, that is in $O(n^2)$ time complexity. The algorithm naturally allows parallel implementation and has an $O(n \log{n})$ reconstruction latency. Furthermore, the \textit{no backtracking condition} of our algorithm is more general than that of the algorithm in \cite{acharya2015string}. This in particular implies that the \textit{reconstruction code} introduced in \cite{pattabiraman2023coding} is reconstructable by our reconstruction algorithm without backtracking. We also improve the time complexity in the case of backtracking. These results are specifically discussed in remark~\ref{rem: time_complexity}, remark~\ref{rem: no_backtracking_time_conclusion}, and remark~\ref{rem: SameBacktracking}. 
\begin{figure}[!htbp]
\centering
{\includegraphics[width=.40\textwidth]{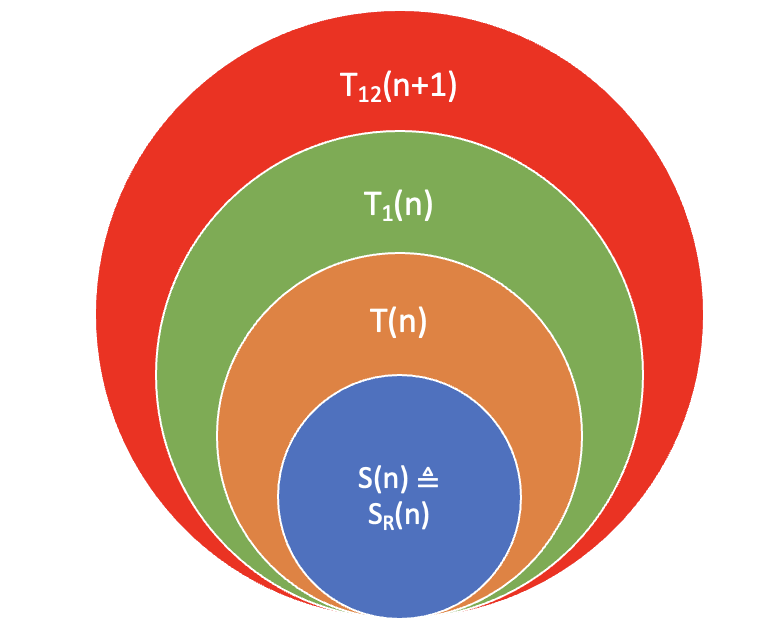}}\\
  \caption{Inclusion relation between different proposed codes and the the previous known code $S_R(n)$.}\vspace{-2mm}
  \label{fig: code_rel}
\end{figure}

In section~\ref{sec: Reconstruction Code}, properties of one-dimensional random walks are leveraged to explicitly characterize the set of binary strings that can be reconstructed by the algorithm in \cite{acharya2015string} without backtracking. In particular, we define this \textit{reconstruction code} to be $S(n)$ and show a bijection between $S(n)$; and $1$-dimensional \textit{positive} $n$-step walks starting from the origin. Using this bijection we propose efficient encoding and decoding procedures for $S(n)$ and show an equivalence between $S(n)$ and the reconstruction code introduced in \cite{pattabiraman2023coding} $S_R(n)$.  We further extend this codebook to propose a new reconstruction code $T(n)$ by expanding codebooks of different sizes in certain specified ways followed by taking a union of them. The size of $T(n)$ is shown to be linearly larger than $S(n)$, and equivalently $S_R(n)$. Furthermore, it is shown that both, the codebook $S(n)$ (and equivalently $S_R(n)$), and the codebook $T(n)$ are reconstructable by our proposed reconstruction algorithm in $O(n^2)$ time. Finally, exploiting the more general sufficient conditions, we slightly modify the proposed algorithm, to introduce larger codebooks $T_1(n)$, and $T_{12}(n)$. A comparison of the rates and redundancies of the different coding schemes is presented in figures~\ref{fig: code_rate} and \ref{fig: code_redundancy}.



The rest of this paper is organized as follows. We describe the problem setting, preliminaries, and relevant previous work in
Section~\ref{sec:Preliminaries}. Then, we describe the new reconstruction algorithm in 
Section~\ref{sec: Reconstruction Algorithm}. In Section~\ref{sec: Reconstruction Code}, we present the new reconstruction code. Finally, we discuss concluding remarks and future research directions in Section~\ref{sec: Conclusion}.

\section{Preliminaries}
\label{sec:Preliminaries}
In this section, we first establish some useful terminology and thereupon formally describe the problem being addressed in this paper. We then recap the results of \cite{acharya2015string}, and \cite{pattabiraman2023coding} which give the relevant background and provide a simple polynomial characterizing of the problem.
\subsection{Problem Formulation}
\label{sec:ProblemFormulation}

Let $s = s_1s_2 \ldots s_n$ be a binary string of length $n \ge 2$ and let $s_{i}^{j}$ denote the contiguous substring $s_is_{i+1} \ldots s_j$ of $s$, where $1\le i \le j \le n$. We will say that a substring $s_i^j$ has the composition $1^w0^z$ where $w$ and $z$ denote the number of $1s$ and $0s$ in the substring respectively. The composition multiset $C(s)$ of a sequence $s$ is the multiset of compositions of all contiguous substrings of $s$. For example, if $s = 1001$, then $C(s) = \{ 0^1,0^1,1^1,1^1,0^11^1,0^11^1,0^2,0^21^1,0^21^1,0^21^2 \}$.

\begin{definition}
\label{def: A(s)}
For a binary string $s$ of length $n$ and weight $d$, let $a_i$ be the number of zeros between the $i^{th}$ and $(i+1)^{th}$ $1$. Define $A(s)$ as the integer string $a_0 a_1 \cdots a_d$. 
\end{definition}

\begin{lemma}\label{lem: bijection}
$ A(s) \to s$ is a bijection between binary strings of length $n$, weight $d$ and non-negative integer strings of length $d+1$, weight (sum of values) $n-d$.
\end{lemma}
\begin{proof}
Consider the mapping that maps the non-negative integer strings of length $d+1$ and weight $n-d$ to binary strings of length $n$ by constructing the corresponding binary string from an $A(s)$ as evident in definition~\ref{def: A(s)}.That is 
$$s =\underbrace{00 \cdots 0}_{a_0} 1 \underbrace{00 \cdots 0}_{a_1}1\underbrace{00 \cdots 0}_{a_2} 1 \cdots 1\underbrace{00 \cdots 0}_{a_d}.$$Now consider two such distinct non-negative integer strings $a = a_0\ldots a_d$ and $b = b_0 \ldots b_d$. If the first position they differ in is $i$, that is $a_i \neq b_i$ and $a_j = b_j$ for $0 \le j \le i-1$, then the corresponding binary strings differ in the positions of their $i^{th}$ 1s. Therefore, each such non-negative integer string corresponds to a unique binary string; implying that the mapping is injective. It is easy to see that both sets have the same size $\binom{n}{d}$, therefore implying the bijection.
\end{proof} 

A set of binary strings of fixed length is called a \textit{reconstruction code} if the composition multisets corresponding to the strings are distinct \cite{pattabiraman2023coding}. Note that a string $s$, and its reverse string ($s^r = s_ns_{n-1} \ldots s_1$) share the same composition multiset and therefore cannot simultaneously belong to a reconstruction code. We restrict the analysis of reconstruction codes to the subsets of strings of length $n$ beginning with $1$ and ending at $0$. This restriction only adds a constant redundancy to the code while ensuring that a string and its reversal are not simultaneously part of the code.

In this paper, the following two problems are addressed (1) Does there exist an efficient algorithm to reconstruct a binary string given its composition multiset?, and (2) Do their exist reconstruction codes of small redundancy and consequently, large rate that can be efficiently encoded and decoded, and can be reconstructed from their composition multiset efficiently? In section~\ref{sec: Reconstruction Algorithm}, we propose a new backtracking algorithm that reconstructs a string $s$ by recovering the integer string $A(s)$ from the corresponding composition multiset $C(s)$. We will use the bijection in lemma~\ref{lem: bijection} to design our reconstruction algorithm, and subsequently in section~\ref{sec: Reconstruction Code} give different families of reconstruction codes that satisfy the aforementioned properties.  
We will also use the following notations in our subsequent proofs: for a string $s$ and the corresponding integer string $A(s) = a_0a_1 \ldots a_d$, we use  $A_i^j(s)$ to denote the substring $a_ia_{i+1} \ldots a_j$ of $A(s)$ and $g_{i}^{j}(s)$ to denote the sum $a_i+a_{i+1} \ldots + a_j$, where $0 \le i \le j \le d$. Whenever clear from the context, \textit{we omit the argument $s$}. Observe that for any string $s$ with weight $d$, $g_0^d = n-d$. For instance, if $s = 10011010$, then $A(s) = 02011$ and $g_1^3 = 3$.



\subsection{Previous Work}
\label{sec: PreviousWork}
In this section, we first review the results of \cite{acharya2015string} that describe the equivalent polynomial formulation of binary strings and their composition multisets. This formulation is central to the design of our~\nameref{Reconstruction Algorithm} which we present in the next section. 
Thereafter, to construct our reconstruction code, we review some elementary results from random walks, and revisit the design of the reconstruction code introduced in \cite{pattabiraman2023coding}. 


\begin{definition}
\label{def: P_s(x,y)}
For a binary string $s =s_1s_2 \ldots s_n$, a bivariate polynomial $P_s(x,y)$ of degree $n$ is defined such that $P_s(x,y) = \sum_{i=0}^n \left(P_s(x,y)\right)_i$, where $\left(P_s(x,y)\right)_0 = 1$ and $\left(P_s(x,y)\right)_i$ is defined recursively as
\begin{equation}
\left(P_s(x,y)\right)_i =    \begin{cases}
y\left(P_s(x,y)\right)_{i-1} &\text{ if $s_i=0$}, \\
x\left(P_s(x,y)\right)_{i-1} &\text{ if $s_i=1$}. \\
\end{cases}
\end{equation}
\end{definition}

$P_s(x,y)$ contains exactly one term of total degree $j$ where $0\le j \le n$ and the coefficient of each term is $1$. The term of the polynomial with degree $j$ is of the form $x^wy^z$ where the substring $s_{1}^{j}$ of $s$ has composition $1^w0^z$. For example, if we consider the string $s=1001$, then $P_s(x,y) = 1+x+xy+xy^2+x^2y^2$. 
 

Similar to the bivariate polynomial for a binary string, we describe a bivariate polynomial $S_s(x,y)$ corresponding to every composition multiset. We associate each element $1^l0^m$ of the multiset with the monomial $x^ly^m$. This is equivalent to saying that an $x$ corresponds to a $1$ and a $y$ corresponds to a $0$ in every monomial of $S_s(x,y)$. As an example, for $s=1001$, $C(s) = \{ 0^1,0^1,1^1,1^1,0^11^1,0^11^1,0^2,0^21^1,0^21^1,0^21^2 \}$ and $S_s(x,y) = 2x+2y+2xy+y^2+2x^2y+x^2y^2$. 

We use the following identity from \cite{acharya2015string}:
\begin{equation}
\label{eqn: acharya_identity}
    P_s\left(x,y\right)P_s \left(\frac{1}{x}, \frac{1}{y} \right) = \left( n+1 \right) + S_s(x,y) + S_s \left(\frac{1}{x}, \frac{1}{y} \right).
\end{equation}

\begin{definition}
\label{def: f*(x,y)}
For a polynomial $f(x,y)$, let $f^*(x,y)$ be the polynomial (also known as reciprocal polynomial) defined as:
\begin{equation}
  f^*(x,y) \ \myeq \ x^{deg_x(f)}y^{deg_y(f)}f \left(\frac{1}{x}, \frac{1}{y} \right).  
\end{equation}
\end{definition}

It is easy to see that $f^*(x,y)$ is indeed a polynomial. 

\begin{remark}
If $P_s(x,y)$ is the bivariate polynomial for the string $s$, then $P_s^*(x,y) = P_{s^r}(x,y)$; that is $P_s^*(x,y)$ is the bivariate polynomial corresponding to the reverse string $s^r = s_ns_{n-1} \ldots s_1$. 
\end{remark}

\begin{definition}
\label{def: F(x,y)}
For a binary string $s$ of length $n$, and the corresponding polynomial $P_s(x,y)$, we define a polynomial $F_s(x,y)$ as: 
\begin{equation}
\label{Fdefinition}
    F_s(x,y) \ \myeq \ P_s(x,y)P_s^*(x,y).
\end{equation}
\end{definition}

Rewriting equation~\eqref{eqn: acharya_identity}, and using the definition in equation~\eqref{Fdefinition}, we obtain
\begin{equation} \label{eqn: F-S transformation}
    F_s(x,y) = x^{deg_x(P_s)}y^{deg_y(P_s)}\left( n+1+ S_s(x,y) \right)  + S_s^*(x,y).
\end{equation}

\begin{remark}\label{rem: F-S transformation}
This result shows that that the polynomial $F_s(x,y)$ can be evaluated directly from $S_s(x,y)$ or equivalently, the composition multiset. 
\end{remark}

\begin{lemma}\label{lemma: F-S equivalence}
For a binary string $s$, the polynomial $F_s(x,y)$ uniquely determines the composition multiset.
\end{lemma}
\begin{proof}
Note that coefficient of $x^ay^b$ in $S(x,y)$ is less than the number of contiguous substrings of length $(a+b)$, which is less than $(n+1)$. Therefore from equation~\ref{eqn: F-S transformation}, $deg_x(P_s)$ and $deg_y(P_s)$ can be uniquely recovered as the degrees of the only monomial with coefficient $\ge (n+1)$. $F_s(x,y)x^{-deg_x(P_s)}y^{-deg_y(P_s)} = (n+1)+S_s(x,y) + S_s(\frac{1}{x},\frac{1}{y})$. Since the polynomial $S_s(x,y)$ has no constant term, the coefficients of $S_s(x,y)$ can be obtained by comparing the coefficients of each degree on both sides of the equality, thereby proving the lemma. 
\end{proof}

\begin{remark}
The remark~\ref{rem: F-S transformation}, and lemma~\ref{lemma: F-S equivalence} imply a bijection between the composition multiset $C(s)$ and the polynomial $F_s(x,y)$. This equivalence is necessary to show the veracity of our algorithm.
\end{remark}

Now, we discuss the preliminaries required for the design of reconstruction code introduced in Section~\ref{sec: Reconstruction Code}. Lemma~\ref{property} gives sufficient conditions for a binary string to be uniquely reconstructed in polynomial time complexity by the algorithm in \cite{acharya2015string}. 

\begin{definition}\label{def: imbalanced_string}
    If a binary string $s$ of length $n$, is such that for all prefix-suffix pairs of length $1 \le j \le n$, one has $wt(s_1^j) \neq wt(s_{n+1-j}^n)$, then $s$ will be called an \textit{imbalanced} string.
\end{definition}

\begin{lemma}[\cite{acharya2015string}, Lemma 37]\label{property}
An imbalanced string $s$ of length $n$ is uniquely reconstructable in $O(n^2 \log{n})$ time by the reconstruction algorithm of \cite{acharya2015string}.
\end{lemma}

In section~\ref{sec: Reconstruction Code}, we show a bijection between \textit{imbalanced} strings of length $n$ that begin with $1$ and end with $0$, and positive $n$-step walks on a line. Using this bijection, we explicitly characterize the set of binary strings reconstructable by the algorithm in \cite{acharya2015string}. 


\begin{definition}\label{def: walk}
A $1$-dimensional \textit{positive} $n$-step walk is defined as an assignment of $n$ variables $X_i \in \{-1,1\}$ for $1 \le i \le n$, such that $S_k = \sum_{i=1}^k X_i$ is positive for $1 \le k \le n$. 
\end{definition}

\begin{lemma}[\cite{feller1967introduction}, Lemma 3.1]\label{lem: NumberOfWalks}
The number of $1$-dimensional \textit{positive} $n$-step walks is $\binom{n-1}{\lfloor \frac{n-1}{2} \rfloor}$.
\end{lemma}

The reconstruction code in \cite{pattabiraman2023coding} uses Catalan-type strings to construct a codebook. The codebook is designed so that for any given codeword and any same-length prefix-suffix substring pair of that codeword, the two substrings have distinct weights.

\begin{definition}[\cite{pattabiraman2023coding}]\label{def: SR(n)}
For reconstruction code $S_R(n)$ of even length ($n$ even):
\[\begin{aligned}
    \mathllap{S_R(n)} \ \myeq \ \{ s \in \{0,1\}^n &\text{, such } \text{that } s_1 = 0 \text{, }s_{n} = 1, \\
    \exists \ I \in \{2,3,&\ldots ,n-1\} \text{, such } \text{that }\\  &\qquad \hspace{.5cm}\text{ for all } i \in I, s_i \neq s_{n+1-i} \\
    &\qquad \hspace{.5cm}\text{ for all } i \notin I, s_i = s_{n+1-i} \\
    &s_{[n/2]\cap I} \text{ is } \text{a Catalan Type String}
    \}.
\end{aligned}\\
\]

For reconstruction code $S_R(n)$ of odd length ($n$ odd):
\begin{align*}
  S_R(n) \ \myeq \ \{ s_1^{(n-1)/2}0s_{(n+1)/2}^{n-1}, s_1^{(n-1)/2}1s_{(n+1)/2}^{n-1}  \text{,} \\ \text{where }s \in S_R(n-1)\}.  
\end{align*}
\end{definition}

The authors in \cite{pattabiraman2023coding} extend this coding scheme to correct single and multiple mass errors. These code extensions relied only upon the fact that all strings in $S_R(n)$ are \textit{imbalanced} strings. In \cite{ye2022reconstruction}, the authors show an equivalence between the set of imbalanced strings beginning with $0$, and ending with $1$, and the codebook $S_R(n)$.

\begin{lemma}[\cite{ye2022reconstruction}, Lemma IV.2]\label{lem: S_R(n)_imbalance_equivalence}
    $S_R(n)$ is the set of all \textit{imbalanced} binary strings of length $n$  beginning with $0$, and ending in $1$.
\end{lemma}


Finally, we give well known bounds on the central binomial coefficient which we will use to show the rate of our reconstruction code.

\begin{lemma}
\label{lemma: 2nChoosen bounds}
The central binomial coefficient may be bounded as:
\begin{align}
  \frac{4^{n}}{\sqrt{\pi (n+1/2)}} \le \binom{2n}{n} \le \frac{4^{n}}{\sqrt{\pi n}} \hspace{1cm} \forall \ n \ge 1.  
\end{align}
\end{lemma}
\vspace{2mm}
\section{Reconstruction Algorithm} \label{sec: Reconstruction Algorithm}
As discussed in Section \ref{sec:ProblemFormulation}, 
we only work with binary strings beginning with $1$ and ending with $0$. In other words, only strings $s= s_1 \ldots s_n$ with $s_1=1,s_n=0$ are considered. In this section, we introduce a new reconstruction algorithm to recover such strings from a given composition multiset. Given a composition multiset, our reconstruction algorithm successively reconstructs $A(s) = a_0 \ldots a_d$, starting from both ends and progressing towards the center. In other words, $a_0$ and $a_d$ are covered first, followed by $a_1$ and $a_{d-1}$, etc.; and the algorithm backtracks when there is an error in recovering a pair. The algorithm takes as input the polynomial $F(x,y)$ (definition~\ref{def: F(x,y)}). Note that the polynomial $F(x,y)$ can be derived from $S(x,y)$ (remark~\ref{rem: F-S transformation}) which in turn is equivalent to the corresponding composition multiset. The algorithm will return the set of strings which have the given composition multiset. We will use the fact that for a string $s$ with the given composition multiset, we must have $F_s(x,y) = F(x,y)$. The lemma~\ref{lemma: F-S equivalence} guarantees that strings recovered in this way indeed have the desired composition multiset.

Before the algorithm is discussed, we first show how certain parameters of a string $s$ with the given composition multiset can be readily recovered from the polynomial $F(x,y)$. These parameters will be subsequently used as inputs to the algorithm. 
\medskip \\
For a string $s= s_1 \ldots s_n$ with $s_1=1,s_n=0$, the corresponding non-negative integer string $A(s)$ (definition~\ref{def: A(s)}) is such that $a_0 = 0$ and $a_d \ge 1$. Using definitions \ref{def: P_s(x,y)} and \ref{def: f*(x,y)},
\begin{align}
    \label{eqn: P_s(x,1)}P_s(x,1) &= 1 + \left(a_1+1 \right)x + \cdots + \left(a_d+1 \right)x^d, \\
    \label{eqn: P_s*(x,1)}P_s^*(x,1) &= (a_d+1) + \cdots + \left(a_1+1 \right)x^{d-1}+x^d.
\end{align}

Since a string $s$ with the given composition multiset must have $F_s(x,y) = F(x,y)$, from definition~\ref{def: F(x,y)}: $F(x,1)=F_s(x,1) = P_s(x,1)P_s^*(x,1)$. Therefore, using equations~\eqref{eqn: P_s(x,1)} and \eqref{eqn: P_s*(x,1)}, the weight of the string $s$ and $a_d$ (where $A(s) = a_0 \ldots a_d$) can be recovered from $F(x,y)$ as follows:

\begin{equation}
    wt(s) = d = \frac{\deg{F(x,1)}}{2} \hspace{1mm}\text{ and } \hspace{2mm}
    a_d = F(0,1) - 1.
\end{equation}

The algorithm will utilize the polynomial formulation of the problem by mapping them to elements of a polynomial ring by considering the coefficients as elements of a sufficiently large finite field, i.e., $\mathbb{F}_q$ with $q$ being a prime number greater than $n$. Let $\lambda \in \mathbb{F}_q$ be a primitive element of this field. We will discuss several properties of the polynomials $P_s(x,\lambda)$ and $P_s^*(x,\lambda)$ (which lie in the ring $\mathbb{F}_q[x]$) which we use in the algorithm. 

\begin{definition}
Given $a_1, \ldots, a_j$ and $a_d,\ldots,a_{d-j}$ in $\mathbb{N} \cup \{0\}$; define the polynomials $\alpha_j(y)$ and $\beta_j(y)$ as follows:
 \begin{align}
     \label{def: polynomial alpha_i}
     \alpha_j(y) &= y^{g_{0}^{j-1}} + y^{1+g_{0}^{j-1}} +\cdots+ y^{g_{0}^{j}},\\
     \label{def: polynomial beta_i}
     \beta_j(y) &= y^{g_{d-j+1}^{d}} + y^{1+g_{d-j+1}^{d}} +\cdots+ y^{g_{d-j}^{d}};
 \end{align}
 where $g_{k}^{l}$ denotes the sum $a_k+a_{k+1} \ldots + a_l$ (defined in section~\ref{sec:ProblemFormulation}).
\end{definition}
Then, using definitions \ref{def: P_s(x,y)} and \ref{def: f*(x,y)}, $\alpha_j(\lambda)$ and $\beta_j(\lambda)$ denote the coefficients of $x^j$ in $P_s(x,\lambda)$ and $P_s^*(x,\lambda)$, respectively.

In particular,
\begin{align}
&\alpha_0(\lambda) = 1, \text{ and } \hspace{1mm} \alpha_d(\lambda) = \lambda^{n-d-a_d} ( 1+\lambda+\cdots+\lambda^{a_d});\\ 
&\beta_{d}(\lambda) = \lambda^{n-d},  \text{ and } \hspace{1mm}   \beta_{0}(\lambda) = 1+\lambda+\cdots+\lambda^{a_d}.
\end{align}

\begin{remark}
$\alpha_j(\gamma)$ and $\beta_j(\gamma)$ correspond to the coefficients of $x^j$ in $P_s(x,\gamma)$ and $P_s^*(x,\gamma)$, respectively, for all $\gamma \in \mathbb{F}_q$. For instance, putting $\gamma=1$ gives $\alpha_j(1) = a_j+1$ and $\beta_j(1) = a_{d-j}+1$ which are the coefficients of $x^j$ in the polynomials $P_s(x,1)$ (equation~\eqref{eqn: P_s(x,1)}) and $P_s^*(x,1)$ (equation~\eqref{eqn: P_s*(x,1)}), respectively.
\end{remark}

The reconstruction algorithm will find $a_j$ and $a_{d-j}$ together at step $j$. Note that in equation~\ref{def: polynomial alpha_i}$, \alpha_i(y)$ is defined using $g_0^i$ and $g_0^{i-1}$, and therefore, can be obtained by knowing the elements $a_1, \ldots, a_i$. Similarly, $\beta_i(y)$ can be obtained from $a_d,\ldots ,a_{d-i}$. Hence, for a string $s$, if by the end of step $j-1$, the algorithm recovers the pairs $(a_1,a_{d-1}),(a_2,a_{d-2}), \ldots ,(a_{j-1},a_{d-j+1})$; the  polynomials $\alpha_0(y), \ldots ,\alpha_{j-1}(y)$ and $\beta_0(y), \ldots, \beta_{j-1}(y)$ are well defined.

\begin{definition}\label{def: f_j(y)}
Let $r_j(y)$ denote the coefficient of $x^j$ in $F(x,y)$. Then $r_j(y)$ can be treated as a polynomial in $y$. At the end of step $j-1$, for polynomials $\alpha_0(y), \ldots, \alpha_{j-1}(y)$ and $\beta_0(y), \ldots, \beta_{j-1}(y)$, define the polynomial $f_j(y)$ as follows:
\begin{equation}
    f_j(y) \ \myeq \ r_j(y) - \sum_{k=1}^{j-1} \alpha_k(y)\beta_{j-k}(y).
\end{equation} 
\end{definition}

By the end of step $j-1$, since we know the polynomials $\alpha_0(y), \ldots, \alpha_{j-1}(y)$ and $\beta_0(y), \ldots, \beta_{j-1}(y)$, we can compute $f_j(y)$. At step $j$, the algorithm wants to find the pair $(a_j,a_{d-j})$. If the pairs $(a_1,a_{d-1}), \ldots \\(a_{j-1},a_{d-j+1})$ are identified correctly, then for the correct pair $(a_j,a_{d-j})$, the coefficient of $x^j$ in $F_s(x,y) \in \mathbb{F}_q[x]$ is $\sum_{i=0}^j\alpha_i(y)\beta_{j-i}(y)$. Since $F_s(x,y) = F(x,y)$, we must have $\sum_{i=0}^j\alpha_i(y)\beta_{j-i}(y) = r_j(y)$. As discussed above, we already know $\alpha_0(y), \ldots ,\alpha_{j-1}(y)$ and $\beta_0(y), \ldots, \beta_{j-1}(y)$ by step $j-1$; therefore a correct pair $(a_j,a_{d-j})$ must satisfy
\begin{equation}\label{eqn: f_j(y)}
f_j(y) = \alpha_0(y)\beta_j(y)+\alpha_j(y)\beta_0(y).    
\end{equation}
By noting that the degrees of both sides should be equal, we have
\begin{equation}\label{eqn: deg(f_j)}
\begin{aligned}
deg(f_j) &= \max \{deg(\alpha_0\beta_j),deg(\alpha_j\beta_0) \} \\
&= \max \{ g_{d-j}^{d},g_{0}^{j}+a_d \}.
\end{aligned}
\end{equation}

Furthermore, observe that $\alpha_i(1) = \beta_{d-i}(1) = 1+a_i$ and hence, $f_j(1) = \beta_j(1) + \left(a_d+1\right)\alpha_j(1)$. From this we obtain: 
\begin{equation}\label{eqn: f_j(1)}
f_j(1) = \left(1+a_{d-j}\right) +  \left(a_d+1\right)\left(a_j+1\right).    
\end{equation} 

We will use these equations to compute the possible values for the pairs $(a_j,a_{d-j})$. Note that equations \eqref{eqn: deg(f_j)} and \eqref{eqn: f_j(1)} give us two possible values for the pair $(a_j,a_{d-j})$ at step $j$. This procedure is captured in the algorithm presented below. The correctness of the algorithm is guaranteed by the lemma~\ref{lemma: F-S equivalence} which showed the strings which share the same $F(x,y)$ indeed share the same composition multiset.

\begin{remark}
    In proposition~\ref{ProblemStringLemma}, we show that verifying the equation \eqref{eqn: f_j(y)} for $y = \lambda$ and $y = \lambda^{-1}$, where $\lambda$ is a primitive root of $\mathbb{F}_q$, is enough to say that equation \eqref{eqn: f_j(y)} holds for all $y$.
\end{remark}


\begin{algorithm}[H] \label{Reconstruction Algorithm}
\SetAlgoLined
\SetKwInOut{Input}{Input}
\SetKwInOut{Output}{Output}

\SetAlgoLined

\vspace{.1cm}
\Input{Polynomial $F(x,y)$, array $A$ of size $d$ initialized with $A[d]=a_d$ and $A[i]=0$ for all $0 \le i \le d-1$,}

\Output{Codestrings $s \in \{0,1\}^n$ }
\vspace{-.3cm}

\hrulefill \\
\vspace{-.4cm}
\DontPrintSemicolon
  \SetKwFunction{FMain}{Reconstruction}
  \SetKwProg{Fn}{Function}{:}{}
  
  \Fn{\FMain{$j$, $F$, $M$}}
  {
    \medskip  
    \uIf{$j = d/2$}
    {
        \uIf{M corresponds to some binary string $s$}{
        S = s \; 
        }
        \uElse{S = empty set \;
        }{\KwRet $S$}
    }
    
    \medskip
    Compute $deg(f_j)$ and $f_j(1)$\;
    
    \medskip    
    $a_j = deg(f_j) - (g_0^{j-1}+a_d)$\\
    \hangindent=0\skiptext\hangafter=0
    $a_{d-j}=f_j(1)-1-(a_d+1)\left(a_j+1\right)$ \;
    \medskip
  
        \uIf{$a_j \ge 0$, $a_{d-j} \ge 0$, $\alpha_0(\lambda)\beta_j(\lambda)+\alpha_j(\lambda)\beta_0(\lambda) = f_j(\lambda)$, and $\alpha_0(\lambda^{-1})\beta_j(\lambda^{-1})+\alpha_j(\lambda^{-1})\beta_0(\lambda^{-1}) = f_j(\lambda^{-1})$}
        { $M[j] = a_j, M[d-j] = a_{d-j}$\\
        \hangindent=0\skiptext\hangafter=0
        $S =$   \FMain{$j+1$, $F$, $M$}\; }
    
    \medskip
    $a_{d-j} = deg(f_j)-g_{d-j+1}^d$ \\
    \hangindent=0\skiptext\hangafter=0
    $a_j = (f_j(1)-1-a_{d-j})/(a_d+1)$\;
    \medskip  
  
        \uIf{$a_j \in \mathbb{N} \cup \{0\}$, $a_{d-j} \ge 0$, $\alpha_0(\lambda)\beta_j(\lambda)+\alpha_j(\lambda)\beta_0(\lambda) = f_j(\lambda)$, and $\alpha_0(\lambda^{-1})\beta_j(\lambda^{-1})+\alpha_j(\lambda^{-1})\beta_0(\lambda^{-1}) = f_j(\lambda^{-1})$}
        { $M[j] = a_j, M[d-j] = a_{d-j}$ \\
        \hangindent=0\skiptext\hangafter=0
        $S = S \ \cup$   \FMain{$j+1$, $F$, $M$}\;
        }
   \vspace{.1cm}
   \KwRet $S$\;
   \vspace{.1cm}
 }
\caption{Reconstruction Algorithm}
\end{algorithm}
\begin{remark}
From equations~\eqref{def: polynomial alpha_i}, and \eqref{def: polynomial beta_i}, we see that $\alpha_k(\lambda)$ and $\beta_{l}(\lambda)$ are of the form $\frac{\lambda^{a}(\lambda^{b}-1)}{\lambda-1}= \lambda^{a+b'}$, where $\frac{\lambda^{b}-1}{\lambda-1} = \lambda^{b'}$. Assuming addition as an $O(1)$ operation, pre-storing $b'$ corresponding to $b$; $\alpha_k(\lambda)\beta_{j-k}(\lambda)$ can be evaluated in $O(1)$ as a power of $\lambda$ and consequently, $\sum_{k=1}^{j-1}\alpha_k(\lambda)\beta_{j-k}(\lambda)$ can be calculated in $O(j)$ time and $O(n)$ space. If the coefficient $a_{k,l}$ of $x^ky^l$ of the polynomial $F(x,y)$ are stored in a matrix, then $a_{j,l}\lambda^l$ can be calculated in $O(1)$ time, and the $n$ row values can be summed in $O(n)$ time. Thus $f_j(\lambda)$ can be calculated in $O(n)$ time.
\begin{remark}
Asymptotically addition is an $O(\log{n})$ process, but for practically relevant values of $n$, addition can be considered an $O(1)$ process. For example, on a $32$-bit system, two $32$ bit numbers can be added in one cycle, and therefore for $\log n < 32$, addition can be assumed to be an $O(1)$ process, and therefore for practical values of $n$, calculating $f_j(\lambda)$ is an $O(n)$ process.
\end{remark}
\begin{remark}
Since the degree and the coefficients of the polynomial $f_j(y)$ (Definition~\ref{def: f_j(y)}) are always non-negative integers less than $n$, $deg(f_j) = \lfloor log_{n+1}(f_j(n+1)) \rfloor$. 
\end{remark}
\end{remark}
\begin{remark}\label{rem: time_complexity}
Assuming no back-tracking, time complexity of the algorithm is $O(dn) = O(n^2)$. This is better than the $O(n^2 \log{n})$ time complexity of the backtracking algorithm proposed by Acharya et. al. in \cite{acharya2015string}. Furthermore, the reconstruction algorithm can be implemented over $O(n \log n)$ latency by executing additions in parallel while calculating $f_j(\lambda)$ etc.
\end{remark}

The reconstruction algorithm has at most two valid choices for the pair $(a_j,a_{d-j})$ at step $j$, and therefore can have at most two branches at any step. If both the conditions are satisfied i.e. both choices are valid according to the algorithm; then our algorithm must choose one direction to proceed. If an error is encountered later, the algorithm comes back to the last branch (not taken yet) where both conditions were satisfied and takes the alternate path. If exactly one condition is satisfied, then our algorithm takes the corresponding path. If neither of the two conditions are satisfied, then assuming the input composition multiset to be valid, our algorithm must have taken the wrong branch in the past (when it had a choice). In such a scenario, our algorithm goes back to the last valid branch where both conditions were satisfied, and takes the alternate branch and proceeds as described. 

We say that a string $s$ \textit{stops} at step $j$ if the algorithm fails to uniquely determine $(a_j,a_{d-j})$ at step $j$. As explained above, this is possible if either both or neither of the two conditions are satisfied. In both cases, the algorithm had a step $j' \le j$ where both of the two conditions were satisfied. Therefore, we will say a string $s$ \textit{pauses} at step $j$ if there are two acceptable branches for $(a_j,a_{d-j})$. In the following lemma, we give algebraic conditions \ref{Type1ProblemString} and \ref{Type2ProblemString}, characterizing the strings that \textit{pause} at some step $j$. 

\begin{proposition} \label{ProblemStringLemma}
Let the bi-variate polynomial corresponding to a string $s$ be  $F_s(x,y)$. Then the reconstruction algorithm \textit{pauses} at step $j$ if and only if the string $s$ satisfies either of the following two relations:
\begin{align}
\label{Type1ProblemString}g_0^{j} - g_{d-j}^d &= a_0+1 = 1 \hspace{3mm} \text{ and } a_j \ge 1\\
\label{Type2ProblemString}g_{d-j}^d - g_0^{j}  &= a_d+1 \hspace{9mm} \ \text{ and } a_{d-j} \ge a_d+1
\end{align}
Moreover, when the reconstruction algorithm \textit{pauses} at step $j$, both the choices for the tuple $(a_j,a_{d-j})$ satisfy equation~\ref{eqn: f_j(y)}.
  
\end{proposition}
\begin{proof}
The algorithm pauses at step $j$, if there exist two pairs of 2-tuples, say $(a_j,a_{d-j})$ and $(a_j',a_{d-j}')$ such that both of them satisfy equation~\eqref{eqn: f_j(y)} for $y=\lambda$ and $y = \lambda^{-1}$. That is, for both these pairs, the corresponding polynomials $(\alpha_j, \beta_j)$ and $(\alpha'_j, \beta'_j)$  respectively satisfy 
\begin{align}
\label{LamdaRelation} \nonumber f_j(\lambda) & = \alpha_0(\lambda)\beta_j(\lambda)+\alpha_j(\lambda)\beta_0(\lambda)  \\ 
&= \alpha'_0(\lambda)\beta'_j(\lambda)+\alpha'_j(\lambda)\beta'_0(\lambda);
\end{align}
\vspace{-10mm}
\begin{align}
\label{LambdaInverseRelation}    \nonumber f_j(\lambda^{-1}) 
 &= \alpha_0(\lambda^{-1})\beta_j(\lambda^{-1})+\alpha_j(\lambda^{-1})\beta_0(\lambda^{-1}) \\
 &=   \alpha'_0(\lambda^{-1})\beta'_j(\lambda^{-1})+\alpha'_j(\lambda^{-1})\beta'_0(\lambda^{-1}).     
\end{align}

By the step $j-1$, we know $g_0^{j-1}$ and $g_{d-j+1}^d$. Using equation~\eqref{eqn: deg(f_j)},
\begin{equation} \label{PauseRelation1}
  g_{d-j+1}^d+a_{d-j}' = deg (f_j) = g_0^{j-1}+a_j+a_d.  
\end{equation}

Using equation~\eqref{eqn: f_j(1)},
\begin{equation}\label{PauseRelation2}
\begin{aligned}
&(1+a_{d-j}) + (a_j+1)(a_d+1) \\ &= (1+a_{d-j}') + (a_j'+1)(a_d+1), \\ 
 \implies & (a_j-a_j')(a_d+1) = (a_{d-j}' - a_{d-j}).
\end{aligned}    
\end{equation}

From equation~\eqref{LamdaRelation},
\begin{align*}
    f_j(\lambda) \ = \ &\lambda^{g_{d-j+1}^d} \left( \sum_{i=0}^{a_{d-j}} \lambda^i \right) + \lambda^{g_{0}^{j-1}}\left( \sum_{i=0}^{a_j} \lambda^i \right)\left( \sum_{i=0}^{a_{d}} \lambda^i \right)\\
    = \ &\lambda^{g_{d-j+1}^d} \left( \sum_{i=0}^{a_{d-j}'} \lambda^i \right) + \lambda^{g_{0}^{j-1}}\left( \sum_{i=0}^{a_j'} \lambda^i \right)\left( \sum_{i=0}^{a_{d}} \lambda^i \right).
\end{align*}
Since $\lambda$ is a primitive element, equating the two expressions and multiplying by $(\lambda-1)^2$,

\begin{equation}\label{LambdaEquation}
\begin{aligned}
    \lambda^{g_{d-j+1}^d}(\lambda-1)(\lambda^{a_{d-j}+1}-\lambda^{a'_{d-j}+1})v \\ 
    =  \lambda^{g_{0}^{j-1}}(\lambda^{a_{d}+1}-1)(\lambda^{a_{j}'+1}-\lambda^{a_{j}+1}).
\end{aligned}    
\end{equation}

Similarly using equation~\eqref{LambdaInverseRelation}, and equating the expressions after multiplying by $(\lambda^{-1}-1)^2$; 
\begin{equation*}
\begin{aligned}
    \lambda^{-g_{d-j+1}^d}(\lambda^{-1}-1)(\lambda^{-a_{d-j}-1}-\lambda^{-a'_{d-j}-1}) \\
    = \lambda^{-g_{0}^{j-1}}(\lambda^{-a_{d}-1}-1)(\lambda^{-a_{j}'-1}-\lambda^{-a_{j}-1}).
\end{aligned}    
\end{equation*}
Simplifying, we get  
\begin{align*}
    &\lambda^{-g_{d-j+1}^d-a_{d-j}-a_{d-j}'-3}(\lambda-1)(\lambda^{a_{d-j}+1}-\lambda^{a'_{d-j}+1})\\ =\ &\lambda^{-g_{0}^{j-1}-a_j-a_j'-a_d-3} (\lambda^{a_{d}+1}-1)(\lambda^{a_{j}'+1}-\lambda^{a_{j}+1}).
\end{align*}
Now using relation~\eqref{LambdaEquation} and equating power of $\lambda$ (which can be done since $\lambda$ is primitive root in a field of size $p>n$)

\begin{equation}\label{PauseRelation3}
    2g_{d-j+1}^d+a_{d-j}+a_{d-j}' = 2g_{0}^{j-1}+a_j+a_j'+a_d.
\end{equation}

For simplicity, we write $ t= a_{d-j} -g_0^{j-1}$. Solving the four equations obtained from~\eqref{PauseRelation1},~\eqref{PauseRelation2}, and~\eqref{PauseRelation3}; we get
    \begin{align}
    \label{firstTuple}(a_j,a_{d-j}) &= (t+g_{d-j+1}^d +1,t + g_0^{j-1}), \\    
    \label{secondTuple} (a_j',a_{d-j}') &= (t+g_{d-j+1}^d,t+ g_0^{j-1}+a_d+1);
    \end{align}
    
where $t = deg(f_j) - ( 1 +a_d +g_{0}^{j-1}+g_{d-j+1}^{d})$.

The tuple $(a_j,a_{d-j})$ corresponds to the condition~\eqref{Type1ProblemString} and the tuple $(a_j',a_{d-j}')$ corresponds to the condition~\eqref{Type2ProblemString}. It is easy to verify that both these tuples indeed give the same polynomial $f_j$ therefore satisfying equation~\ref{eqn: f_j(y)}. Furthermore, this relation implies that $t$ is unique and if the reconstruction algorithm pauses at step $j$, then there are exactly two choices for the tuple.
\end{proof}

\begin{remark}
The above proof implies that if the algorithm \textit{pauses} at step $j$, then the two valid solutions for $(a_j,a_{d-j})$ exactly correspond to the conditions~\ref{Type1ProblemString} and ~\ref{Type2ProblemString}.  
\end{remark}

\begin{definition}\label{def: problem_strings}
We will call the strings which satisfy condition~\eqref{Type1ProblemString} for some $0<j<d/2$ as \textit{type-1 strings}, and the strings which satisfy condition~\eqref{Type2ProblemString} for some $0<j<d/2$ as \textit{type-2 strings}. 
\end{definition}

\begin{remark}\label{rem: WhenReconstruct}
A string can be a \textit{type-1 string}, a \textit{type-2 string}, both a \textit{type-1} and a \textit{type-2 string}, or be of neither type. Since our algorithm can only confuse a \textit{type-1 string} with a \textit{type-2 string}, if our algorithm knows the type of string, it can know which branch to choose thereby avoiding backtraking. In section~\ref{sec: Reconstruction Code}, we will use this fact to design reconstruction codes by avoiding all strings as a single type to be given as input.
\end{remark}

\begin{corollary} \label{cor: S(n)IsReconstructable}
If an \textit{imbalanced} string $s$ (definition~\ref{def: imbalanced_string}) of length $n$ is such that it begins in $1$ and ends at $0$,  then $s$ can be uniquely reconstructed in $O(n^2)$ time.
\end{corollary}

\begin{proof}
We will show that an \textit{imbalanced} string cannot be a \textit{type-1 string}. As discussed in the previous remark, telling our algorithm to always choose condition~\ref{Type2ProblemString} in case of a \textit{pause}, any such string can be reconstructed without backtracking and hence in $O(n^2)$ time. \\ 
Let if possible, $s$ also be a \textit{type-1 string}. Let step $j$ be the first time the string $s$ \textit{pauses} and satisfies  condition~\ref{Type1ProblemString}. If condition~\ref{Type1ProblemString} is satisfied, then the $j^{th}$ one in $s$ is at position $g_0^j + j$, and the $j^{th}$ last one in $s$ is at position $g_0^j + j-1$ from the end of the string. Therefore, 
\begin{equation}\label{eqn: negative_weight}
  wt \left( s_1^{g_0^j+j-1} \right) - wt\left( s_{n+1-g_{d-j}^d-j}^d \right) = (j-1) - j = -1.  
\end{equation}
But note that $wt(s_1^1) - wt(s_n^n) = 1 $. Consider the function $f(i) = wt(s_1^i) - wt(s_{n-i+1}^n)$. This function is such that $f(i+1) = f(i) \pm 1$. Therefore, the function must have been zero at some point, contradicting the fact that $s$ is \textit{imbalanced}.
\end{proof}

\begin{remark}\label{rem: no_backtracking_time_conclusion}
Corollary~\ref{cor: S(n)IsReconstructable} implies that our algorithm uniquely reconstructs the codewords of the reconstruction code $S_R(n)$ described in \cite{pattabiraman2023coding}  (revisited in section \ref{sec: PreviousWork}) without backtracking. In remark~\ref{rem: time_complexity} we showed the reconstruction algorithm presented in this paper has a worst-case time complexity of $O(n^2)$ when there is no backtracking compared to the reconstruction algorithm in \cite{acharya2015string} which has a time complexity of $O(n^2\log{n})$ .
\end{remark}

\begin{remark}
\label{rem: SameBacktracking}
If we define $I$ and $l_s$ as defined in \cite{acharya2015string}, that is
\begin{align*}
&I \ \myeq \   \{ i<n/2: w(s_1^i)=w(s_{n+1-i}) \text{ and } s_{i+1} \neq s_{n-i} \}, \\ &\text{ and } l_s = |I|;
\end{align*}
by proof of corollary~\ref{cor: S(n)IsReconstructable}, each time the string s \textit{pauses} at some step $j$, we have $i \in I$ with $i<j$. Therefore the number of branches in case of backtracking in our algorithm is less than or equal to $l_s$ which is the number of branches of the backtracking algorithm in \cite{acharya2015string}. Thus our algorithm is able to find $s$ before depth $l_s + 1$ and therefore for practical values of $n$, the time complexity of our algorithm is $O(2^{l_s}n^2)$ compared to the algorithm in \cite{acharya2015string} whose time complexity is $O(2^{l_s}n^{2}\log{n})$.
\end{remark}

\section{Reconstruction Code}\label{sec: Reconstruction Code}
In this section, we explicitly describe the reconstruction code $S(n)$ (definition~\ref{def: S(n)}) which will consist of all \textit{imbalanced} strings (definition~\ref{def: imbalanced_string}) of length $n$, beginning with $1$, and ending at $0$. The design of our reconstruction code is such that we avoid all strings satisfying condition ~\ref{Type1ProblemString} in our codebook. This will ensure that in case of a \textit{pause}, the reconstruction algorithm will know which branch to take. For a string to not be uniquely reconstructable, it must \textit{pause} at some step; therefore, avoiding \textit{pauses} ensures that the string is uniquely reconstructed from its composition multiset. Note that lemma~\ref{lem: S_R(n)_imbalance_equivalence} implies that the reconstruction code $S_R(n)$ (definition~\ref{def: SR(n)}) is the reverse of the reconstruction code $S(n)$. We show a bijection between $S(n)$ and positive $n$-step walks (defintion~\ref{def: walk}) thereby explicitly describing the code size and propose efficient procedures for mapping information message into this code and then retrieving them. The bounds on the redundancy are provided in corollary~\ref{cor: S(n)_redundancy}. Corollary~\ref{cor: S(n)IsReconstructable} ensures that the elements of $S(n)$ are uniquely reconstructable by our~\nameref{Reconstruction Algorithm} in $O(n^2)$ time. Recall that the elements of this codebook $S(n)$ are also reconstructable by the algorithm in \cite{acharya2015string} without backtracking (lemma~\ref{property}). The relevant background for this section is discussed in section~\ref{sec: PreviousWork}.\\
Later, we extend $S(n)$ by expanding codebooks of different sizes in certain specified ways followed by taking a union of them, in order to arrive at a new codebook $T(n)$. This codebook $T(n)$ contains $S(n)$, but also has strings that are not $\textit{imbalanced}$. The more general sufficient conditions for reconstruction in polynomial time of our algorithm (proposition~\ref{ProblemStringLemma}) ensure that elements of the codebook $T(n)$ can be reconstructed in $O(n^2)$ time. Finally, using the ideas discussed in remark~\ref{rem: WhenReconstruct}, we propose codebooks $T_1(n)$, and $T_{12}(n)$, through which we give computational bounds on the size of reconstruction codebooks uniquely reconstructable by the reconstruction algorithm in $O(n^2)$ time.

\begin{definition}\label{def: S(n)}
Define $S(n)$ to be the set of all \textit{imbalanced} binary strings of length $n$ beginning with $1$, and ending at $0$; that is for all prefix-suffix pairs of length $1 \le j \le n$, one has $wt(s_1^j) \neq wt(s_{n+1-j}^n)$.
\end{definition}
\begin{theorem}\label{thm: bijection}
There is a bijection between $S(n)$ and positive $n$-step walks (definition~\ref{def: walk}).
\end{theorem}
\begin{proof}
Given a binary string $s= s_1 \ldots s_n$, assign $X_i$'s in the following way:

\begin{equation*}
    X_{2i-1} = \begin{cases} \hspace{3mm}1, \text{ if } s_i = 1, \\ -1, \text{ if } s_i = 0; \end{cases} \hspace{4mm}\text{and } \hspace{4mm}
    X_{2i} = \begin{cases} -1, \text{ if } s_{n-i} = 1, \\ \hspace{3mm} 1, \text{ if } s_{n-i} = 0. \end{cases}
\end{equation*}
This assignment is uniquely invertible. That is, for each such $s$, there is a unique assignment of variables $X_j$ and vice versa. Now note that, $S_{2k} = \sum_{i=1}^{2k}X_i = 2(wt(s_1^k) - wt(s_{n-k+1}^n))$. Therefore, $S_{2k} = 0 \iff wt(s_1^k) = wt(s_{n-k+1}^n)$ implying the required bijection.
\end{proof}
\vspace{2mm}
The above result along with lemma~\ref{lem: NumberOfWalks} gives us the following corollary.

\begin{corollary}\label{cor: S(n)_redundancy}
The size of $S(n)$ is given by $ \binom{n-1}{\lfloor \frac{n-1}{2} \rfloor} \ge \frac{2^{n-\frac{1}{2}}}{\sqrt{\pi n}} $. Therefore, redundancy of the reconstruction code $S(n)$ is at most $\lceil 1/2 \log{n} + 1/2 + 1/2 \log_2{\pi} \rceil$.
\end{corollary}

\begin{remark}\label{rem: prefix-suffix-positive}
From the above proof, it is easy to see that for a string $s \in S(n)$, $(wt(s_1^k) - wt(s_{n-k+1}^n)) > 0$ for $1 \le k \le n/2$.
\end{remark}

The bijection in theorem~\ref{thm: bijection} also gives us a way of explicitly constructing the reconstruction code $S(n)$ i.e. mapping and retrieving information messages from the codebook elements. In the book \cite{feller1967introduction}, a $1$-dimensional random walk is interpreted as a "mountain range" with upstrokes, and downstrokes. Formally, for an assignment of $n$ variables $X_i \in \{-1,1\}$ for $1 \le i \le n$, the $1$-dimensional random walk is mapped to a lattice path beginning from origin, with the $i^{th}$ step size as $(1,X_i)$. Note that this construction maps a positive random walk to a lattice path that always stays above the $x$-axis. The book then provides a recipe to geometrically map a path from $(0,0)$ to $(2n,0)$ into $2n$ length paths from $(0,0)$ with all vertices strictly above or on the axis. That is, the positive $n$ step walks are explicitly mapped to the size of code $S(n)$ which is $\binom{n-1}{\lfloor \frac{n-1}{2} \rfloor}$ as stated in corollary~\ref{cor: S(n)_redundancy}. This mapping when merged with the bijection in the proof of theorem~\ref{thm: bijection} can be adapted to give us a procedure to explicitly bijectively map \textit{imbalanced} strings beginning with $1$ and ending at $0$ to the process of selection of some $\lfloor \frac{n-1}{2} \rfloor$ objects from $(n-1)$ objects. \\
In \cite{kabal2018combinatorial}, the author uses a coding trellis to give an efficient way of encoding/decoding combinatorial indices for a selection of items from a given set of items. Therefore, a combination of the procedures described in \cite{feller1967introduction}, and \cite{kabal2018combinatorial} can be used to define a map from integers in $[0,\binom{n-1}{\lfloor \frac{n-1}{2} \rfloor}-1]$ to the set of \textit{imbalanced} strings beginning with $1$, and ending at $0$. 
\vspace{2mm}\\

\begin{remark}
The result from \cite{ye2022reconstruction}, discussed in section~\ref{sec: PreviousWork} as lemma~\ref{lem: S_R(n)_imbalance_equivalence}, shows that the codebook $S_R(n)$ is the reverse of the codebook $S(n)$.   
\end{remark}
\begin{remark}
In \cite{ye2022reconstruction}, the authors show that the elements of the codebook $S(n)$ are also uniquely reconstructable from the multiset of their prefix-suffix compositions.
\end{remark}

Now, we finally extend our reconstruction code $S(n)$ by expanding codebooks of different sizes in certain specified ways followed by taking a union of them, in order to arrive at a new codebook $T(n)$. We define the following kinds of sets whose construction uses this $S(n)$. The reconstruction code $T(n)$ will be defined as the union of these sets. 

\begin{definition}\label{def: P_nk}
Given a positive integer $n$, and $2 \le k \le \lfloor n/2 \rfloor$, define $P_{k}$ as a set of binary strings of length $n$ which begin at $1$, end at $0$, as follows:
\begin{align}
\nonumber
    P_{k} = \{s \in \{0,1\}^n, \ t \in \{0,1\}^{k-2} \text{, such } \text{that }  & \ \ \\ \nonumber s_1^{k} =  1t0\text{, }s_{n-k+1}^{n} = 1t^{r}0 &  \\ 
      \qquad \hspace{.5cm}\text{ and } s_{k+1}^{n-k} \in S(n-2k) \}&,
\end{align}
where $t^r$ denotes the reverse of the string $t$.
\end{definition}

\begin{proposition}Given a binary string $s \in P_{k}$ of length $n$, with $2 \le k \le \lfloor n/2 \rfloor$; $s$ is uniquely reconstructable by our algorithm.
\end{proposition}

\begin{proof}
We will show that any $s \in P_{k}$ is not a \textit{type-1 string}, and therefore the result will from remark~\ref{rem: WhenReconstruct}. This proof will be similar to the proof of corollary~\ref{cor: S(n)IsReconstructable}. Consider the function $f(i) = wt(s_1^i) - wt(s_{n-i+1}^n)$. Note that, $$f(i)  \begin{cases} = 1, \text{ for } 1\le i \le k-1; \\ = 0, \text{ for }i = k; \\ >0, \text{ otherwise.} \end{cases}$$
The first two results follow from the construction of $P_{k}$ in definition~\ref{def: P_nk}, and the last inequality follows from remark~\ref{rem: prefix-suffix-positive}.
As seen in the equation~\ref{eqn: negative_weight}, in the proof of corollary~\ref{cor: S(n)IsReconstructable}; for every \textit{type-1 string}, there exists a $j'$, such that $f(j') = -1$, implying that $s \in P_{k}$ cannot be a  \textit{type-1 string}.    
\end{proof}

\begin{definition}\label{def: T(n)}
Define $T(n) = S(n) \bigcup \left(\bigcup_{k=1}^{n/2-1} P_{k}\right)$.
\end{definition}

\begin{remark}
    The extended codebook presented in the ISIT 2022 version of this paper \cite{gupta2022new} avoided \textit{type-2 strings} and was shown to be larger than $S(n)$ by a linear factor $41/40$. The codebook defined here avoids \textit{type-1 strings} and is shown to be larger than $S(n)$ be a linear factor of $9/8$.
\end{remark}

\begin{theorem}
Given $\epsilon>0$, there exists an $N \in \mathbb{N}$ such that for all integers $n>N$ we have 
\begin{equation}
  |T(n)| \ge \hspace{.2mm} \left(1.125 - \epsilon \right)|S(n)|.  
\end{equation}
\end{theorem}

\begin{proof}
Let $s_1 \in P_{k_1}$, and $s_2 \in P_{k_2}$ with $k_1 \neq k_2$. Then $$wt(s_{1}^{k_1}) - wt(s_{n+1-k_1}^{n}) = 0 \neq wt(s_{1}^{k_2}) - wt(s_{n+1-k_2}^{n}).$$ This means that $P_{k_1} \cap P_{k_2} = \Phi$. Now note that, 
\begin{align*}
    \frac{|T(n)|}{|S(n)|} = 1 + \sum_{k=2}^{\lfloor n/2 \rfloor} \frac{|P_{n,k}|}{|S(n)|} = 1+\sum_{k=2}^{\lfloor n/2 \rfloor} 2^{k-2}\frac{\binom{n-1-2k}{\lfloor \frac{n-1}{2}\rfloor - k}}{\binom{n-1}{\lfloor \frac{n-1}{2}\rfloor}}.
\end{align*}
Setting $n = 2n'+1$, we see that,
\begin{align*}
    \frac{|T(n)|}{|S(n)|} &= 1 + \sum_{k=2}^{n'} 2^{k-2} \frac{\binom{2n'-2k}{n'-k}}{\binom{2n'}{n'}} \\
    &= 1 + \sum_{k=2}^{n'} 2^{k-2} \left( 4^{-k} + \frac{k \cdot 4^{-k}}{2n'} + O(\frac{1}{n'^2})\right) \\
    &=1 + \frac{2n'(2^{n'}-3) + 3 \cdot 2^{n'}}{n'\cdot2^{n'+4}} + O(\frac{1}{n}) \\
    &\ge \frac{9}{8} + O(\frac{1}{n}).
\end{align*}
\end{proof}

As we discuss in remark~\ref{rem: WhenReconstruct}, our algorithm can only confuse a \textit{type-1 string} with a \textit{type-2 string}. Therefore, let $T_1(n)$ be the set of all binary strings of length $n$ beginning with $1$, and ending with $0$ with no \textit{type-1 strings}; that is all strings satisfying condition~\ref{Type1ProblemString} for any $1 \le j \le d$ are removed from the set of strings being considered. This means that the set $T_{1}(n)$ contains strings which are either only \textit{type-2}, or neither of the types. Then, for each element in $T_1(n)$, our algorithm even in case of a $\textit{pause}$ knows exactly which branch to take (the branch satisfying condition~\ref{Type2ProblemString}). Therefore, it uniquely reconstructs the string without backtracking, that is in $O(n^2)$ time complexity. \\
Extending this argument further, we define $S_{12}(n)$ to be the set of all binary strings of length $n$ beginning with $1$, and ending with $0$ with no strings that are both \textit{type-1} and \textit{type-2}. That is, all strings satisfying condition~\ref{Type1ProblemString} for some $1 \le j_1 \le d$, and satisfying condition~\ref{Type2ProblemString} for some $1 \le j_2 \le d$ are removed from the set of strings being considered. This means that the set $S_{12}(n)$ contains strings which are either only \textit{type-1}, only \textit{type-2}, or neither of the types. Note that, for our algorithm to know which branch to take, we will need to add an extra bit of redundancy, an indicator bit, to the elements of $S_{12}(n)$. This bit will indicate if the string being considered is \textit{type-2} or not. If the added bit is $1$, in case of a \textit{pause}, our algorithm will know that the string is \textit{type-2}, and take the branch corresponding to condition~\ref{Type2ProblemString}. If the added bit is $0$, in case of a \textit{pause}, our algorithm will know that the string is \textit{type-1}, and take the branch corresponding to condition~\ref{Type1ProblemString}, or continue without backtracking in the case of no \textit{pauses}. We define $T_{12}(n+1)$ to be the codebook of length $(n+1)$ where the codebook is formed by adding this indicator bit to the elements of $S_{12}(n)$. Note that, by construction we have the following relationship between the proposed codebooks, also represented in figure~\ref{fig: code_rel}:
\begin{equation}
    S(n) \subset T(n) \subset T_1(n) \subset T_{12}(n+1).
\end{equation}



\begin{figure}[H]
\centering
\begin{minipage}{.5\textwidth}
  \centering
  \includegraphics[width=.9\linewidth]{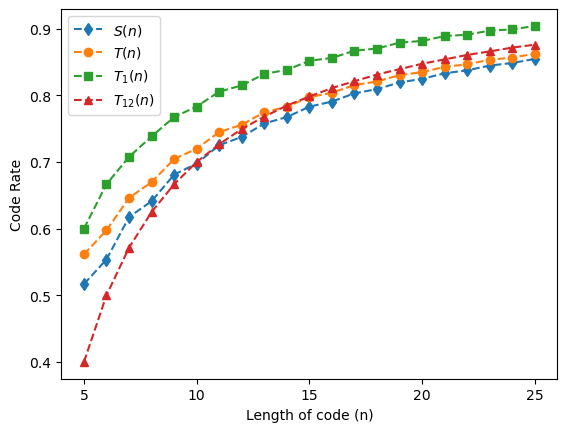}
  \caption{Comparison of code rates}\vspace{-2mm}
  \label{fig: code_rate}
\end{minipage}%
\begin{minipage}{.5\textwidth}
  \centering
      \includegraphics[width=.9\linewidth]{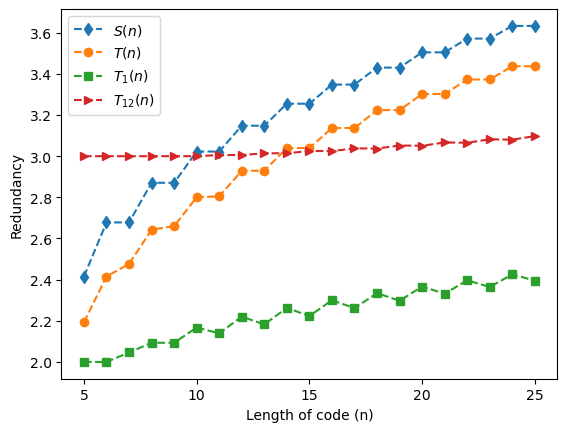}
   \caption{Comparison of code redundancies}\vspace{-2mm}
  \label{fig: code_redundancy}
\end{minipage}
\end{figure}

\begin{remark}
In figures~\ref{fig: code_rate} and ~\ref{fig: code_redundancy}, we present the code rates and code redundancies of the reconstruction codes $S(n)$ (defintion~\ref{def: S(n)}), T(n) (defintion~\ref{def: T(n)}), and codes $T_1(n)$, and $T_{12}(n)$ as described above. The time complexity for computing $T_{1}(n)$ and $T_{12}(n)$ is $O(n \cdot 2^n)$, and therefore the results are presented only for $n \le 25$.
\end{remark}



\section{Conclusion.}\label{sec: Conclusion}
Motivated by the problem of recovering polymer strings from their fragmented ions during mass spectrometry, we introduce a new algorithm to reconstruct a binary string from the multiset of its substring compositions. This algorithm takes a new algebraic approach, thereby improving the time complexity of reconstruction in the case of no backtracking to $O(n^2)$, as well as in cases where backtracking is needed. We further characterize algebraic properties of binary strings that guarantee reconstruction without backtracking thereby enlarging the space of binary strings uniquely reconstructable without backtracking compared with previously known algorithms. Additionally, we modify and extend the reconstruction code proposed in \cite{pattabiraman2019reconstruction} to produce a new reconstruction code which is linearly larger in size, and is uniquely reconstructable by our algorithm without backtracking. 

There are several combinatorial and coding-theoretic problems related to string reconstruction from substring composition that remain open. The problems of bounding the size of \textit{reconstruction codes} as well as constructing explicit schemes with minimum redundancy remain open. Our algorithm expands the conditions for strings to be uniquely reconstructed without backtracking, and therefore characterizing the set of strings uniquely reconstructable by the algorithm in this paper is a possible step in that direction. As seen from results in figure~\ref{fig: code_redundancy}, we believe that there exist reconstruction codes with constant redundancy that can be reconstructed efficiently. Furthermore, deriving bounds on time complexity of algorithms for reconstructing strings from their substring multiset is another problem of interest.     



\bibliographystyle{IEEEtran}
\bibliography{bibliography}



\end{document}